\documentclass[a4paper,12pt]{amsart}
\usepackage{amsthm}
\usepackage{amsmath,amssymb,amsthm}
\usepackage{amsaddr}
\usepackage{amsfonts}
\usepackage{braket}
\usepackage{chngcntr}
\usepackage{subcaption}
\counterwithin{figure}{section}
\usepackage[T1]{fontenc}
\usepackage{hyperref}
\usepackage{tikz}
\usepackage{graphicx}
\newtheorem{thm}{Theorem}
\newtheorem{lem}{Lemma}
\newtheorem{defn}{Definition}
\newtheorem{prop}{Proposition}
\newtheorem{rem}{Remark}
\theoremstyle{definition}
\newtheorem{defn2}{Example}
\newenvironment{defnbis}[1]
  {\renewcommand{\thedefn}{\ref{#1}$'$}%
   \addtocounter{defn}{-1}%
   \begin{defn}}
  {\end{defn}}
\newcommand{\Complex}{\mathbb{C}}
\newcommand{\Rational}{\mathbb{Q}}

\newcommand{\tree}{%
  \vcenter{\hbox{\tikz[node distance=2.5ex]{%
     \draw[thick] (5,-0.20) -- (5,0) -- (4.85,0.10) -- (5,0) -- (5.15,0.10) ;
}}}}
\newcommand{\bridge}{\raisebox{1.75mm}{$\frown$}\hspace{-4.5mm}\vee}

\begin{document}
\title{Hopf Algebras and Topological Recursion}
\author{Jo\~{a}o N. Esteves}
\address{CAMGSD, Departamento de Matem\'{a}tica, Instituto Superior T\'{e}cnico, Av. Rovisco Pais 1, 1049-001 Lisboa, Portugal}
\email{joao.n.esteves@tecnico.ulisboa.pt}
\thanks{The author was supported by Funda\c{c}\~ao para a Ci\^{e}ncia e a Tecnologia through the grant SFRH/BPD/77123/2011. He also wishes to thank Jos\'{e} Mour\~ao and Nicolas Orantin for fruitful discussions.}
\begin{abstract}
We consider a model for topological recursion based on the Hopf Algebra of planar binary trees of Loday and Ronco. We show that extending this Hopf Algebra by identifying pairs of nearest neighbor leaves and thus producing graphs with loops we obtain the full recursion formula of Eynard and Orantin.
\end{abstract}

\maketitle
\tableofcontents

\section{Introduction}The use of graphs, in particular of trees, binary trees and planar binary trees, in mathematical physics has a long tradition. The canonical examples are perhaps Feynman diagrams but the connection with Hopf Algebras of trees started with the works of Connes and Kreimer \cite{MR1725011,MR1748177,MR1810779} that describe the combinatorics of the procedure of extracting sub-divergences in Quantum Field Theory known as the BPHZ renormalization procedure \cite{collins1984renormalization}. Another approach to the use of graphs in QFT and in particular in QED, considering binary trees, planar or not, was followed by Brouder and Frabetti \cite{Brouder:1999gk,Brouder:1999za,MR1817703}. Later it was understood that these two approaches are very similar and in some cases equivalent and are related to quasi-symmetric and noncommutative quasi-symmetric functions, see for instance \cite{MR2194965,MR1905177,MR1909461,MR1327096}.

 In this paper we show how the Hopf Algebra of planar binary trees of Loday and Ronco \cite{MR1654173} can be seen as a representation of the vector space generated by correlation functions that obey the Eynard-Orantin recursion formula. These correlation functions are graded by the Euler characteristic and we can consider for each degree the vector space over $\Rational$ generated by them and then take the direct sum of these vector spaces for all degrees. 
 First we consider planar binary trees of order $n$, that is with $n$ vertices and $n+1$ leaves, as a representation of genus $g=0$ correlation function $W_{k}^0(p,p_1,\dots, p_{k-1})$ of Euler characteristic $\chi=2-2g-k$ equal to $-n$. Here the Euler characteristic is the one of Riemann or topological surfaces of genus $g$ and $k$ punctures or borders to which the correlation functions $W_k^g(p,p_1,\dots, p_{k-1})$ are usually related in some concrete problems. For $g=0$ we label the root with $p$ and the $n+1$ leaves with the $p_1,\dots,p_{n+1}$ variables. Then by connecting the nearest neighbors leaves with a single edge and reducing the number of pairs of labels in the same way as increasing the genus we obtain graphs with loops that we see as a representation of higher genus correlation functions with the same Euler characteristic. 
 
 As an example take $W^2_1(p)$ which has $\chi=-3$. Its underline generating trees are planar binary trees of order 3 which are also models for $W_5^0(p,p_1,p_2,p_3,p_4)$:
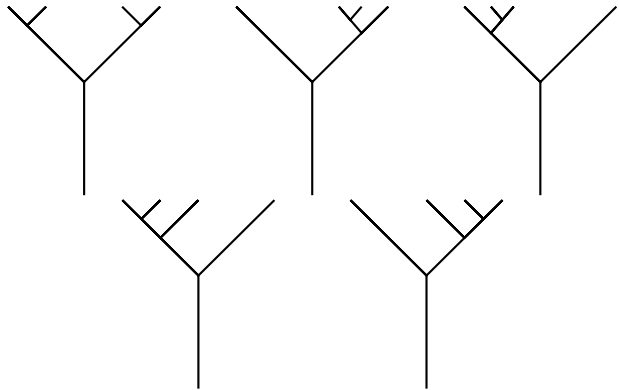
\begin{figure}
  \begin{tikzpicture}\label{fig:pbt5a}
  \draw[thick](1,-0.5) -- (1,1) -- (0.25,1.75) -- (0.5,2)  -- (0.25,1.75) -- (0,2) -- (1,1) -- (2,2) -- (1.75,1.75) -- (1.5,2);\draw[thick] (4,-0.5) -- (4,1) -- (3,2) -- (4,1) -- (5,2) -- (4.65,1.65) -- (4.35,2)-- (4.5,1.825)--(4.65,2);\draw[thick] (7,-0.5) -- (7,1) -- (6,2) -- (6.35,1.65)--(6.65,2)--(6.5,1.825)-- (6.35,2)--(6.5,1.825)--(6.35,1.65)-- (7,1) -- (8,2);
  \end{tikzpicture}
  \begin{tikzpicture}\label{fig:pbt5b}
  \draw[thick] (10,-0.5) -- (10,1) -- (9,2) -- (9.25,1.75)--(9.5,2)--(9.25,1.75) --  (9.5,1.5) -- (10,2) -- (9.5,1.5) -- (10,1) -- (11,2);
  \draw[thick] (13,-0.5) -- (13,1) -- (12,2) -- (13,1) -- (14,2)--(13.75,1.75)--(13.5,2)--(13.75,1.75) --  (13.5,1.5) -- (13,2) -- (13.5,1.5) ;
  \end{tikzpicture}
  \caption{Planar binary trees of order 3 as generators of the correlation functions $W^0_5, W^1_3$ and $W^2_1$ with $\chi=-3$.}\label{fig:pbt5}
\end{figure}
 Identifying pairs of nearest neighbor leaves in the left and right branches independently we get the second term of topological recursion. Identifying pairs of leaves each taken from the left and the right branches gives the first term. Note that in this case not every planar binary tree of order 3 gives $W_1^2$. In fact the first tree of fig. \ref{fig:pbt5} does not give a genus 2 correlation function by identifying the nearest neighbor leaves in opposite branches.
 \section{The topological recursion of Eynard and Orantin}
 The topological recursion formula of Eynard and Orantin has its origin in Matrix Models, for general reviews see for instance \cite{DiFrancesco:1993nw,MR2346575}. In the hermitian 1-matrix form of the theory the purpose is to compute connected correlation functions $W_{k+1}$ depending on a set of variables $p, p_1,\dots,p_k$
 \begin{equation}
 W_{k+1}(p,p_1,\dots p_k)=\Braket{\text{Tr}\frac{1}{p-M}\text{Tr}\frac{1}{p_1-M}\dots \text{Tr}\frac{1}{p_k-M}}_c
 \end{equation}
 starting with $W_1(p)$ and $W_2(p,p_1)$.
  These functions which are solutions of the so-called loop equations are only well defined over Riemann surfaces because in $\Complex$ they are multi-valued. They admit an expansion on the order $N$ of the random matrix $M$, with components $W_{k+1}^g(p,p_1,\dots p_k)$ related to a definite genus. We will not be concerned here with the actual computation of correlation functions in specific models.
  
  Let $K=(p_1,\dots,p_k)$ be a vector of variables. For instance in concrete cases these can be coordinates of punctures on Riemann surfaces, labels of borders on topological surfaces or variables in Matrix Models, but we just leave them as labels of leaves of planar binary trees or of graphs obtained from planar binary trees. We assign the label $p$ to the root of a tree or of a graph with loops obtained from a tree. The topological recursion formula is
 \begin{align}\label{toprec}
 &W_{k+1}^g(p,K)=\sum_{\text{branch points }\alpha}\text{Res}_{p\rightarrow \alpha}K_p(q,\bar{q})\notag\\
 &\left(W^{g-1}_{k+2}(q,\bar{q},K)+\sum_{L\cup M=K,h=0}^g W^h_{|L|+1}(q,L)W^{g-h}_{|M|+1}(\bar{q},M)\right)
 \end{align}
 where the sum is restricted to terms with Euler characteristic equal or smaller than 0. For instance if $h=0$ then $|L|\ge 1$. 
 For a very clear exposition about this setup from the point of view of Algebraic Geometry see for instance \cite{MR3087960} but some comments are in order. The branch points are the ones from a meromorphic function $x$ defined on a so called spectral curve $\mathcal{E}(x,y)=0$. The recursion kernel $K_p(q,\bar{q})$ is, roughly speaking, a meromorphic (1,1) tensor that depends on a regular point $p$ in the neighbourhood of a branch point and on $q$ and its conjugated point $\bar{q}$ for which $x(q)=x(\bar{q})$ and $y(q)=-y(\bar{q})$. In fact it can be computed from $W_1^0(p)$ and $W_2^0(p,p)$ which are symmetric differentials of order one and two respectively. Actually, all $W_k^g$ are meromorphic symmetric differentials but we will continue to refer to them as correlation functions. Since our approach will be purely algebraic and in order to soften the notation we will not explicitly mention the sum of the residues over the branch points when referring to this formula.
 \section{The Loday-Ronco Hopf Algebra of planar binary trees}
 We collect here some important facts of the Loday-Ronco Hopf algebra. Details and proofs can be found in \cite{MR2194965,MR1654173}.
 Let $S_n$ be the symmetric group of order $n$ with the usual product $\rho\cdot\sigma$ given by the composition of permutations. When necessary we denote a permutation $\rho$ by its image $(\rho(1)\rho(2)\dots\rho(n))$. Recall that a shuffle $\rho(p,q)$ of type $(p,q)$ in $S_n$ is a permutation such that $\rho(1)<\rho(2)<\dots <\rho(p)$ and $\rho(p+1)<\rho(p+2)<\dots <\rho(p+q)$. For instance the shuffles of type $(1,2)$ in $S_3$ are $(123),(213)$ and $(312)$. We denote the set  of $(p,q)$ shuffles by $S(p,q)$.
 Take 
 \begin{equation}
 k[S^\infty]=\oplus_{n=0}^{\infty}k[S_n]
 \end{equation}
 with $S_0$ identified with the empty permutation. $k[S^\infty]$ is a vector space over a field $k$ of characteristic $0$ generated by linear combinations of permutations. It is graded by the order of permutations and $k[S_0]$ which contains the empty permutation is identified with the field $k$. For two permutations $\rho\in S_p$ and $\sigma\in S_q$ there is a natural product on $S^\infty$ denoted by $\rho\times\sigma$ which is a permutation on $S_{p+q}$ given by letting $\rho$ acting on the first $p$ variables and $\sigma$ acting on the last $q$ variables.

 There is a unique decomposition of any permutation $\sigma\in S_n$ in two permutations $\sigma_i\in S_i$ and $\sigma'_{n-i}\in S_{n-i}$ for each $i$ such that
 \begin{equation}
 \sigma =(\sigma_i\times\sigma'_{n-i})\cdot w^{-1}
 \end{equation}
 where $w$ is a shuffle of type $(i,n-i)$.
 With the $\ast$ product
 \begin{equation}
 \rho\ast \sigma=\sum_{\alpha_{n,m}\in S_{(n,m)}}\alpha_{n,m}\cdot\left(\rho\times\sigma\right)
 \end{equation}
 and the co-product
 \begin{equation}\label{eq:coproductperm}
  \Delta\sigma=\sum \sigma_{i}\otimes\sigma^{'}_{n-i}
  \end{equation}
  $k[S^{\infty}]$ becomes a bi-algebra and since it is graded and connected it is automatically a Hopf Algebra.

 A planar binary tree is a graph with no loops embedded in the plane with only trivalent vertices. In every planar binary tree there are paths that start on a special edge called the root and end on the terminal edges called leaves. The leaves can be left or right oriented. The order $|t|$ of a planar binary tree $t$ is the number of its vertices and on each planar binary tree of order $n$ there are $n+1$ leaves that usually are numbered from 0 to $n$ from left to right. It is frequent to visualize planar binary trees from the bottom to the top, with the root as its lowest vertical edge and the leaves as the highest edges, oriented SW-NE or SE-NW. We will denote the set of planar binary trees of order $n$ by $Y^n$ and by $k[Y^\infty]$ the vector space over $k$ generated by planar binary trees of all orders.
\begin{figure}
   \begin{tikzpicture}
  \draw[thick] (1,-0.5) -- (1,1) -- (0,2) -- (0.35,1.65)--(0.65,2)--(0.35,1.65)-- (1,1) --(1.65,1.65)--(1.35,2)--(1.65,1.65) -- (2,2);
     \end{tikzpicture}
  \caption{A planar binary tree of order 3}\label{fig:pbt3-2}
\end{figure}
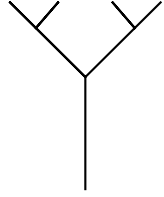
 Additionally a planar binary tree with levels is a planar binary tree such that on each horizontal line there is at most one vertex.
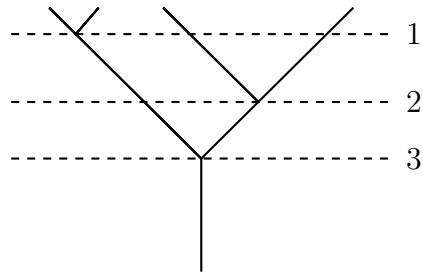
\begin{figure}
   \begin{tikzpicture}
     \draw[thick] (4,-0.5) -- (4,1) -- (2,3) -- (2.35,2.65)--(2.65,3)--(2.35,2.65)-- (4,1) --(4.75,1.75)--(3.5,3)-- (4.75,1.75)--(6,3); \draw[thick,dashed](1.5,1.75)--(6.5,1.75);\draw[thick,dashed](1.5,2.65)--(6.5,2.65);
     \draw (6.8,2.65) node{1}; \draw (6.8,1.75) node{2};\draw[thick,dashed](1.5,1)--(6.5,1);\draw (6.8,1) node{3};
    \end{tikzpicture}
  \caption{Planar binary tree with levels that is the image of $\mathbf{(132)}$}\label{fig:pbtlev3}
\end{figure}
 It is clear that reading the vertices from left to right and from top to bottom it is possible to assign a permutation of order $n$ to a planar binary tree with levels and that this assignment is unique. For example in fig. \ref{fig:pbtlev3} the tree corresponds to the permutation $(132)$. In this way it is completely equivalent to consider the Hopf algebra $k[S^\infty]$ or the Hopf algebra of planar binary trees with levels because they are isomorphic. However Loday and Ronco show in \cite{MR1654173} that the $\ast$ product and the co-product are internal on the algebra of planar binary trees which is then isomorphic to a sub-Hopf algebra of $k[S^\infty]$ with the same product and co-product. The identity of the Hopf Algebra $k[Y^\infty]$ is the tree with a single edge and no vertices, following the convention of considering only internal vertices, which represents the empty permutation, and the trivial permutation of $S_1$ is represented by the tree with one vertex and two leaves, see fig \ref{fig:idgen}. In fact this element is the generator of the augmented algebra by the $\ast$ product. See fig. (\ref{fig:1star1star1}) for an example of an order 3 product.
 
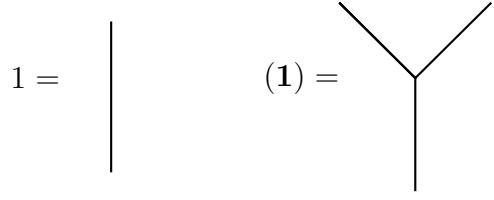
\begin{figure}
  \begin{tikzpicture}
  \draw (0,1) node{ $\Large{1=}$};\draw[thick] (1,-0.25) -- (1,1.75);
  \draw (3.5,1) node{ $\Large{(\mathbf{1})=}$};\draw[thick] (5,-0.5) -- (5,1) -- (4,2) -- (5,1) -- (6,2) ;
  \end{tikzpicture}
  \caption{The identity and the generator in $k[Y^\infty]$}\label{fig:idgen}
\end{figure}
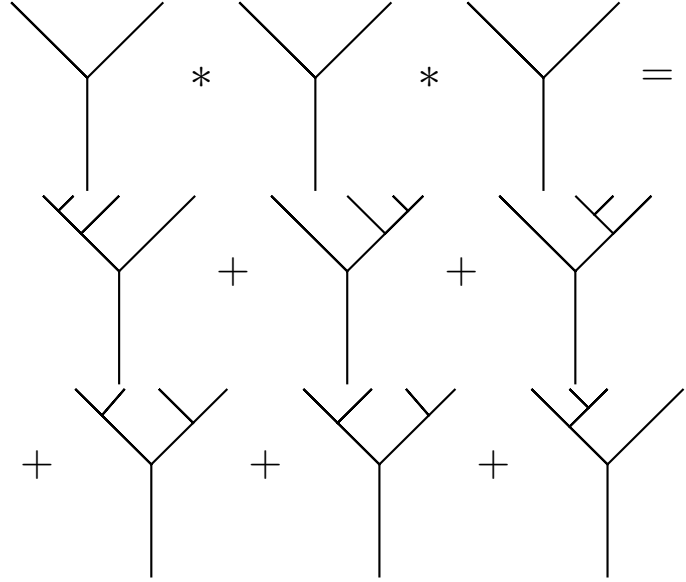
\begin{figure}
  \begin{tikzpicture}
   \draw[thick](1,-0.5) -- (1,1) -- (0,2) -- (1,1) -- (2,2) ; \draw (2.5,1) node{\textbf{{\Large $\ast$}}}; \draw[thick](4,-0.5) -- (4,1) -- (3,2) -- (4,1) -- (5,2); \draw (5.5,1) node{\textbf{{\Large $\ast$}}};\draw[thick](7,-0.5) -- (7,1) -- (6,2) -- (7,1) -- (8,2);\draw (8.5,1) node{\textbf{{\Large $=$}}};
   \end{tikzpicture}
   \begin{tikzpicture}
   \draw[thick] (1,-0.5) -- (1,1) -- (0,2) -- (0.2,1.8)--(0.4,2)--(0.2,1.8)-- (0.5,1.5)--(1,2)--(0.5,1.5)-- (1,1) -- (2,2);\draw (2.5,1) node{\textbf{{\Large $+$}}}; \draw[thick] (4,-0.5) -- (4,1) -- (3,2) -- (4,1) -- (5,2)--(4.8,1.8)--(4.6,2)--(4.8,1.8)--(4.5,1.5)--(4,2);\draw (5.5,1) node{\textbf{{\Large $+$}}}; \draw[thick] (7,-0.5) -- (7,1) -- (6,2) -- (7,1) -- (8,2)--(7.5,1.5)--(7.25,1.75)--(7.5,2)--(7.25,1.75)--(7,2);
    \end{tikzpicture}
    \begin{tikzpicture}
       \draw (-0.5,1) node{\textbf{{\Large $+$}}};\draw[thick] (1,-0.5) -- (1,1) -- (0,2) -- (0.35,1.65)--(0.65,2)--(0.35,1.65)-- (1,1) --(1.55,1.55)--(1.10,2)--(1.55,1.55) -- (2,2);\draw (2.5,1) node{\textbf{{\Large $+$}}};  \draw[thick] (4,-0.5) -- (4,1) -- (3,2) -- (3.45,1.55)--(3.9,2)--(3.45,1.55)-- (4,1) --(4.65,1.65)--(4.35,2)--(4.65,1.65) -- (5,2);\draw (5.5,1) node{\textbf{{\Large $+$}}}; \draw[thick] (7,-0.5) -- (7,1) -- (6,2) -- (6.5,1.5)--(6.75,1.75)--(6.5,2)--(6.75,1.75)--(7,2)-- (6.5,1.5) -- (7,1)-- (8,2);
       \end{tikzpicture} 
\caption{$\mathbf{(1)}\ast\mathbf{(1)}\ast\mathbf{(1)}=\mathbf{(123)}+\mathbf{(321)}+\mathbf{(312)}+\mathbf{(132)}+\mathbf{(231)}+\mathbf{(213)}$ computed in $k[S^\infty]$. Note that in $k[Y^\infty]$ the fourth and the fifth trees are the same.} \label{fig:1star1star1}
 \end{figure}
 The grafting $t_1 \vee t_2$ of two trees $t_1$ and $t_2$ is the operation of producing a new tree $t$ by inserting $t_1$ on the left and $t_2$ on the right leaves of $\mathbf{(1)}$. It is clear that any tree of order $n$ can be written as $t_1\vee t_2$ with $t_1$ of order $p$, $t_2$ of order $q$ and $n=p+q+1$. If a tree has only leaves on the right branch besides the first leaf then it can be written as $|\vee t_2$ and reciprocally if it has only leaves on the left branch besides the last leaf. Note that in particular $\mathbf{(1)}=|\vee |$. 
 In \cite{MR1654173} Loday and Ronco show that the $\ast$ product restricted to planar binary trees satisfies the identity
\begin{equation}\label{eq:shuffleident1}
    t\ast t'=t_1\vee (t_2\ast t')+(t \ast t_1^{'})\vee t_2^{'}
\end{equation}
and
\begin{equation}\label{eq:shuffleident2}
    t\ast |=|\ast t=t
\end{equation}
with $t=t_1\vee t_2$ and $t'=t_1^{'}\vee t_2^{'}$.

 If $t_1$ is a tree of order $p$ and a representative element of $W_{p+2}^{0}(p,L)$ and $t_2$ is a tree of order $q$ and a representative element of $W_{q+2}^{0}(p,M)$ then $t=t_1\vee t_2$ is a tree of order $n=p+q+1$ and a representative element of $$W^0_{n+2}(p,K)=K_p(q,\bar{q})W_{p+2}^{0}(q,L)W_{q+2}^{0}(\bar{q},M).$$ with $K=\{p_1,\dots,p_{n+1}\}=L\cup M$.
 We will clarify this in what follows.
 
  \section{The solution of topological recursion}
  \subsection{Genus 0}
  A representation map $\psi$ from the vector space of correlation functions of genus $g$ to the vector space of graphs with loops should be defined such that in particular to a correlation function $W^0_{n+2}(p,p_1,\dots,p_{n+1})$ of Euler characteristic $\chi=-n$ would correspond the trees of order $n$. In fact we will state below that the representation of $W^0_{n+2}(p,p_1,\dots,p_{n+1})$ is the sum of all trees of order $n$. It is not clear that this map gives a true representation in the strict mathematical sense. It is linear by definition and is obvious that it is surjective, as we can associate some instance of a correlation function $W^0_{n+2}(p,p_1,\dots,p_{n+1})$ to any tree of order $n$. If it is injective and a homomorphism is a more delicate issue because even if one considers $W^0_{n+2}$ as being a sum of all instances of correlation functions of Euler characteristic $-n$ each represented by a tree $t\in Y^n$ in the same way as in Particle Physics, where different Green functions contribute to the same scattering amplitude, it is not evident that the space of correlation functions has a product with an identity that would correspond to the trivial tree |. Note that this would give at least a ring structure and in the case of topological quantum field theory where correlations functions are identified with topological surfaces with punctures cobordism is a good candidate for such a product. In fact it is a consequence of the axioms of topological quantum field theory as stated by Atiyah for example in \cite{atiyah1988topological} that the cylinder $\Sigma\times I$, where $\Sigma$ is a topological surface without border and $I$ is a interval of real numbers, may be identified with the identity map between two vector spaces. In any case we will not elaborate more on this here and use the word representation in a somewhat rough sense. In particular, when referring to the inverse image of a tree or a sum of trees we will refrain of using the inverse $\psi^{-1}$ but will use instead $\psi^\ast$ as for the pullback.
  \begin{defn}\label{def:W3}
  Consider the planar binary tree with one vertex $(\mathbf{1})$. The 3-point correlation function $W_3^0(p,p_1,p_2)$ is represented by the sum of two planar binary trees with one vertex, obtained by the permutation of the leaf labels $p_1$ and $p_2$:
  \begin{equation}
  \psi\left(W_3^0(p,p_1,p_2)\right)=\sum_{\text{perm. of leaf labels $\{p_1,p_2\}$}}(\mathbf{1})
  \end{equation}
 \end{defn}
 
  The trees that represent $W_3^0(p,p_1,p_2)$ are given by the permutations of the leaf labels of $|\vee |$. Then it is natural to represent the operation of grafting two trees by the insertion of the recursion kernel $K_p(q,\bar{q})$ on its roots. Therefore the symbol
 $\tree $ has two meanings. When isolated it represents $W_3^0(p,p_1,p_2)$ because the two cylinders $W^0_2(q,p_1)$ and $W^0_2(\bar{q},p_2)$ are implicitly identified with its leaves.
 When it is an internal vertex of a more complex tree it is the recursion kernel $K_p(q,\bar{q})$ with suitable labels of its variables.
   \begin{defnbis}{def:W3}\label{def:W3b}
      The propagator or cylinder (also named Bergman kernel in the literature) $W^0_2(q,\bar{q})$ is represented through $\psi$ by the empty permutation $|$ and the recursion kernel is represented through $\psi$ by $\tree$ when in an internal vertex of some tree. Then each planar binary tree of order n is a representation of an instance of some correlation function in genus 0 with each vertex identified with a recursion kernel and each left leaf identified with the cylinder $W^0_2(q_i,p_j)$ or each right leaf identified with the cylinder $W^0_2(\bar{q}_i,p_k)$. Finally the image under $\psi$ of a correlation function $W_{n+2}^0(p,p_1,\dots,p_{n+1})$ with $\chi=-n$ is the sum of all planar binary trees of order $n$ considering all permutations of their leaf labels and with the identifications mentioned above,
      \begin{equation}
      \psi\left(W_{n+2}^0(p,p_1,\dots,p_{n+1})\right)=\sum_{\substack{t_i\in Y^n \\ \text{perm. of leaf labels $\{p_1,\dots,p_{n+1}\}$}}} t_i.
      \end{equation}
    \end{defnbis}
    Hence Definition \ref{def:W3} becomes the following example:
    \begin{defn2}\label{ex:W3}
     Consider the planar binary tree with one vertex. The 3-point correlation function $W_3^0(p,p_1,p_2)$ is represented by the sum of two planar binary trees with one vertex, obtained by the permutation of the leaf labels $p_1$ and $p_2$.
     \end{defn2}
    \begin{align}
     \psi\left(W_3^0(p,p_1,p_2)\right)&=\psi\left(K_p(q,\bar{q})W^0_2(q,p_1)W^0_2(\bar{q},p_2)\right)+ \text{ perm. of $\{p_1,p_2\}$}\notag\\
     &=\sum_{\text{ perm. of $\{p_1,p_2\}$}}|\vee |\notag\\
      &=\sum_{\text{ perm. of $\{p_1,p_2\}$}}(\mathbf{1})
    \end{align}
     
   \begin{prop}
    If $W_{n+2}^0(p,p_1,\dots,p_{n+1})$ is a correlation function with Euler characteristic $\chi=-n$ that is a solution of (\ref{toprec}) then we have
    \begin{align}\label{eq:defcorr}
    \psi \left(W_{n+2}^0(p,p_1,\dots,p_{n+1})\right)&=\sum_{\substack{p+q+1=n\\|t_1|=p, |t_2|=q}} t_1\vee t_2\notag\\
    & + \text{perm. of leaf labels $\{p_1,\dots,p_{n+1}\}$}
    \end{align}
    \end{prop}
 \begin{proof}
 This is the topological recursion in genus $0$ written with planar binary trees. For $n=1$ this is the example \ref{ex:W3}. For $n$ arbitrary by Definition \ref{def:W3b} 
 \begin{equation}
       \psi\left(W_{n+2}^0(p,p_1,\dots,p_{n+1})\right)=\sum_{\substack{t\in Y^n \\ \text{perm. of leaf labels $\{p_1,\dots,p_{n+1}\}$}}} t
       \end{equation}
Decompose uniquely any $t$ of order $n$ into $t=t_1\vee t_2$ of orders $|t_1|=p$ and $|t_2|=q$ with $p+q+1=n$ to get
 \begin{equation}\label{eq:invimage}
       \psi\left(W_{n+2}^0(p,p_1,\dots,p_{n+1})\right)=\sum_{\substack{t_1\in Y^p, t_2\in Y^q, p+q+1=n \\ \text{perm. of leaf labels $\{p_1,\dots,p_{n+1}\}$}}} t_1\vee t_2.
       \end{equation}
Then $t_1$ and $t_2$ are on the image by $\psi$ of $W_{p+2}$ and $W_{q+2}$ for $p$ and $q$ varying from 0 to $n-1$ and constrained by $p+q+1=n$. Since the operation of grafting two trees is represented by attaching the recursion kernel to its roots then, summing for all $t_1\in Y^p, t_2\in Y^q$ and for $p+q+1=n$, we get the topological recursion formula for $g=0$ after taking the preimage of (\ref{eq:invimage}) by $\psi$:
 \begin{align}
 W_{n+2}^0(p,p_1,\dots,p_{n+1})&=\notag\\
 K_p(q,\bar{q})&\sum_{\substack{L\cup M=\{p_1,\dots,p_{n+1}\},\\|L|=p+1,|M|=q+1}} W^0_{|L|+1}(q,L)W^{0}_{|M|+1}(\bar{q},M).
 \end{align} 
 \end{proof}
 \begin{rem}
Note that by $W^0_{n+2}(p,K)$ with $|K|=n+1$ we understand all instances of the correlation function with $g=0$ and $n+2$ labels. This is similar to the situation in High Energy Physics where for the same physical process described by a scattering amplitude there are several Feynman diagrams that contribute.

 It is well known that the dimension of the vector space generated by planar binary trees of order $n$ is given by the Catalan number (see for instance \cite{MR1817703}) $$c_n=\frac{2n!}{n!(n+1)n!}.$$ It is also known that correlation functions in Matrix Models have a large $N$ or planar expansion that is given in terms of Catalan numbers. Therefore it is of no surprise that there exists a correspondence between planar binary trees and correlation functions of genus 0.
 \end{rem}

 \begin{thm}
The $n$-order solution $W_{n+2}^0(p_1,\dots,p_{n+1})$ of the topological recursion in genus 0 is represented by the linear combination
$$\sum t=\mathbf{(1)}\ast\mathbf{(1)}\ast\dots\ast\mathbf{(1)}$$
with $n$ factors of $\mathbf{(1)}$ followed by the sum over all permutations of its labels. 
\end{thm}
In this way $W_{n+2}^0(p,p_1,\dots,p_{n+1})$ is represented by $\sum t$ followed by the identification of cylinders $W^0_2(q_i,p_j)$ with the left leaves or $W^0_2(\bar{q}_i,p_k)$ with the right leaves and finally by summing over all permutations of the labels $p_1,p_2,\dots, p_{n+1}$. In other words, the $\ast$ product $\mathbf{(1)}\ast\mathbf{(1)}\ast\dots\ast\mathbf{(1)}$ gives all possible insertions of recursion kernels of $W_{n+2}^0$.
\begin{proof}
By induction on the Euler characteristic or equivalently on the order $n$. For $n=1$ we saw that $W_3^0(p,p_1,p_2)$ is just a sum of two planar binary trees with one vertex, with the leaves in correspondence with $W^0_2(q,p_1)$ and $W^0_2(\bar{q},p_2)$ or its permutations and the root labeled by $p$. 
So we start the induction at $n=2$: we want to show that $(\mathbf{1})\ast(\mathbf{1})$ represents $W_4^0(p,p_1,p_2,p_3)$. In the Loday-Ronco Hopf Algebra we have that $\mathbf{(1)}\ast\mathbf{(1)}=\mathbf{(12)}+\mathbf{(21)}$.
 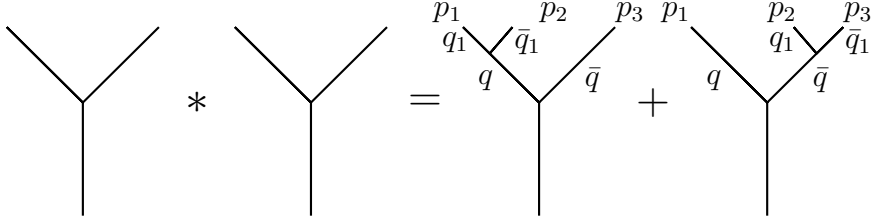
\begin{figure}
 \begin{tikzpicture}
  \draw[thick](1,-0.5) -- (1,1) -- (0,2) -- (1,1) -- (2,2) ; \draw (2.5,1) node{\textbf{{\Large $\ast$}}}; \draw[thick](4,-0.5) -- (4,1) -- (3,2) -- (4,1) -- (5,2);\draw (5.5,1) node{\textbf{{\Large $=$}}};\draw[thick] (7,-0.5) -- (7,1) -- (6,2) -- (6.35,1.65)--(6.65,2)--(6.35,1.65)-- (7,1) -- (8,2);\draw (8.5,1) node{\textbf{{\Large $+$}}}; \draw[thick] (10,-0.5) -- (10,1) -- (9,2) -- (10,1) -- (11,2)--(10.65,1.65)--(10.35,2)--(10.65,1.65);\draw (6.3,1.3) node{ $q$};\draw (7.7,1.3) node{ $\bar{q}$};\draw (9.3,1.3) node{ $q$};\draw (10.7,1.3) node{ $\bar{q}$};
  \draw (5.9,1.8) node{ $q_1$};\draw (6.85,1.8) node{ $\bar{q}_1$};\draw (5.8,2.2) node{ $p_1$}; \draw (7.2,2.2) node{ $p_2$};\draw (8.2,2.2) node{ $p_3$};\draw (10.2,1.8) node{ $q_1$};\draw (11.2,1.8) node{ $\bar{q}_1$};\draw (8.8,2.2) node{ $p_1$}; \draw (10.2,2.2) node{ $p_2$};\draw (11.2,2.2) node{ $p_3$};
  \end{tikzpicture}
 \caption{$W_4^0(p,p_1,p_2,p_3)$}\label{fig:W40}
 \end{figure}
 On the other hand the topological recursion formula gives for $W_4^0$
 \begin{align}
 W_4^0(p,p_1,p_2,p_3)&=K_p(q,\bar{q})\left(W_3^0(q,p_1,p_2)W_2^0(\bar{q},p_3)\right.\notag\\
 &\left.+W_2^0(q,p_1)W_3^0(\bar{q},p_2,p_3)+\text{perm. of }\{p_1,p_2,p_3\}\right)\notag\\
 &=K_p(q,\bar{q})K_q(q_1,\bar{q}_1)W_2^0(q_1,p_1)W_2^0(\bar{q}_1,p_2)W_2^0(\bar{q},p_3)\notag\\
 &+K_p(q,\bar{q})K_{\bar{q}}(q_1,\bar{q}_1)W_2^0(q_1,p_2)W_2^0(\bar{q}_1,p_3)W_2^0(q,p_1)\notag\\
 &+\text{perm. of }\{p_1,p_2,p_3\}
 \end{align}
 which gives the two terms from the $\ast$ product of the Loday-Ronco Hopf Algebra identifying the vertices with the recursion kernel and the leaves with the cylinders (see fig. \ref{fig:W40}).

 Next assume the induction hypothesis for $n-1$ and note that if $W^0_{n+1}$ is represented by the linear combination $\sum t=\mathbf{(1)}\ast\mathbf{(1)}\ast\dots\ast\mathbf{(1)}$ of trees $t\in Y^{n-1}$ then each tree can be written uniquely as $t=t_1\vee t_2$ with $|t_1|=a, |t_2|=b$ and $a+b+1=n-1$. Using (\ref{eq:shuffleident1}) and (\ref{eq:shuffleident2}) we get
 \begin{align}\label{eq:proof}
   \sum_{t\in Y^{n-1}} \mathbf{(1)}\ast  t &= \sum_{t\in Y^{n-1}}|\vee (|\ast t) + \sum_{\substack{t_1\in Y^{a},t_2\in Y^b\\ a+b+1=n-1}}(\mathbf{(1)}\ast t_1)\vee t_2\notag\\
    &=\sum_{t\in Y^{n-1}}|\vee t + \sum_{\substack{t_1\in Y^{a},t_2\in Y^b\\ a+b+1=n-1}}(\mathbf{(1)}\ast t_1)\vee t_2
 \end{align}
 Because each $t$ is of order $n-1$ but otherwise arbitrary, each $t_1$ is at most of order $n-2$ and then by the induction hypothesis $\mathbf{(1)}\ast t_1$ is in the image by $\psi$ of a solution that is at most $W^0_{n+1}$ and at least $W_3^0$. Summing also over all permutations of $K=\{p_1,\dots p_{n+1}\}$ we get
 \begin{align}\label{eq:proff2}
  \psi^\ast\left(\sum_{\substack{t\in Y^{n-1}\\\text{perm. of leaf labels}}}\mathbf{(1)}\ast t\right) &= \sum_{L\cup M=K, |L|=1} K_p(q,\bar{q}) W_2^0(q,L) W_{n+1}^0(\bar{q},M) \notag\\
 +\sum_{L\cup M=K, |L|>1} K_p(q,\bar{q})& W_{l+1}^0(q,L) W_{m+1}^0(\bar{q},M)\notag\\
 =\sum_{L\cup M=K} K_p(q,\bar{q}) & W_{l+1}^0(q,L) W_{m+1}^0(\bar{q},M)\notag\\
 =W_{n+2}^0&(p,p_1,\dots,p_{n+1}).
\end{align}
\end{proof}
\subsection{Genus higher than 0}
The procedure of attaching an edge to two consecutive leaves and producing a graph with loops allows to represent correlations functions with genus $g>0$. This is equivalent to extract the outermost cylinders $W^0_2(x,p_j),W_2^0(y,p_{j+1}), x=q_i$ or $\bar{q}_i$, $y=\bar{q}_j$ or $q_j$ and to couple a cylinder $W^0_2(x,y)$  to two recursion kernels $K_{q_l}(q_i,\bar{q}_i)$ and $K_{q_m}(q_j,\bar{q}_j)$, for some convenient choice of indices, that are identified with two internal vertices. This procedure does not change the Euler characteristic of the associated correlation functions because the number of pairs of leaf labels is reduced exactly as the genus is increased. For instance with this procedure we can make the sequence
\begin{equation}
W_5^0(p,p_1,\dots p_4)\longrightarrow W^1_3(p,p_1,p_2)\longrightarrow W^2_1(p)
\end{equation} 
and remain in the same graded vector space that contains $k[Y^3]$. How this changes the Hopf algebra structure is not yet clear. For now, we define the operation $_{i}\leftrightarrow_{i+1}$ on a planar binary tree. 
\begin{defn}\label{defn:connecting}
Starting with a planar binary tree of order $n$ and $n+2$ labels (including the root label $p$) the operation $_{i}\leftrightarrow_{i+1}$ consists in erasing the labels of the leaves $i$ and $i+1$ then connecting them by an edge and finally relabeling the remaining leaves, now numbered $j$ with $j=0,\dots,n-2$, with the $p_{j+1}$ labels, producing in this way a graph with one loop.
\end{defn}
Therefore we represent a correlation function $W^g_{k}(p,p_1,\dots,p_{k-1})$ of genus $g$ by graphs with loops $t^g$ that are obtained by successive applications of the $_{i}\leftrightarrow_{i+1}$ operation. We denote by $\left(Y^n\right)^g$ the set of different graphs with $g$ loops that are obtained from trees $t\in Y^n$.
\begin{defn}
A correlation function $W^g_{k}(p,p_1,\dots,p_{k-1})$ of genus $g$ and Euler characteristic $\chi=2-2g-k$ is represented by a sum of all different graphs with loops $t^g\in \left(Y^n\right)^g$ for  $n=-\chi$:
\begin{equation}
\psi\left(W^g_{k}(p,p_1,\dots,p_{k-1})\right)=\sum_{t^g\in (Y^n)^g} t^g
\end{equation}
\end{defn}
\begin{rem}
Two graphs $(t)^g,(t')^g\in (Y^n)^g$ are considered different in the obvious way. Either the underlying binary trees $t,t'\in Y^n$ are distinct as base elements of $k[Y^\infty]$ or the tree $t$ has a pair of leaves, say $(i,i+1)$ that are identified with an edge in $(t)^g$ producing a loop and are free in $(t')^g$ (and reciprocally because the two graphs have the same genus). 
\end{rem}
In particular $W_1^1(p)$ is represented by a single graph with one loop denoted $(\mathbf{1})^1$ whose underlying planar binary tree is $(\mathbf{1})$. More generally we have
\begin{prop}\label{prop:secterm-W1}
The second summand of the topological recursion formula for the correlation function $W_{n}^1(p,p_1,\dots,p_{n-1})$ with $\chi=-n$ is represented by the sum 
\begin{equation}\label{eq:prop2}
\sum_{\substack{(t_1)^1\in \left(Y^p\right)^1, t_2\in Y^q\\p+q+1=n}}(t_1)^1 \vee t_2 +\sum_{\substack{t_1\in Y^p, (t_2)^1\in \left(Y^q\right)^1\\p+q+1=n}} t_1\vee (t_2)^1
\end{equation}
where each underlying planar binary tree $t$ of order $n$ is decomposed as $t=t_1\vee t_2$, with $|t_1|=p$, $|t_2|=q$, $p+q+1=n$.
\end{prop}
\begin{proof}
For $n=2$ the underlying trees of $W_2^1(p,p_1)$ are the same of $W^0_4(p,p_1,p_2,p_3)$ namely $t=(\mathbf{12})+(\mathbf{21})$. Remembering that $(\mathbf{12})=(\mathbf{1})\vee|$ and $(\mathbf{21})=|\vee(\mathbf{1})$ and applying $_{i}\leftrightarrow_{i+1}$ to $(\mathbf{1})$ for $i=0$ (\ref{eq:prop2}) is the same as
$$
(\mathbf{1})^1\vee | + |\vee (\mathbf{1})^1.
$$
which is the image by $\psi$ of the second term of $W_2^1(p,p_1)$ in the topological recursion formula.
For $n$ arbitrary consider the trees $t'=t'_1\vee t'_2, t'\in Y^{n-1}$ with $|t'_1|=a, |t'_2|=b, a+b+1=n-1$ whose sum represents solutions $W^0_{n+1}(p,p_1,\dots,p_n)$ by the induction hypothesis.  By (\ref{eq:shuffleident1}) and (\ref{eq:shuffleident2}) we have 
\begin{equation}\label{eq:shuffleident1a}
(\mathbf{1})\ast t'= |\vee t' + ((\mathbf{1})\ast t'_1)\vee t'_2.
\end{equation}
Noting that $t^{''}=(\mathbf{1})\ast t'_1$  are at most of order $n-1$ and at least of order $1$ and identifying pairs of leaves on the same branches by applying $_{i}\leftrightarrow_{i+1}$ to each component of the grafting operation the last formula gives
\begin{equation}\label{eq:identifysepbr}
\left(_{i}\leftrightarrow_{i+1}\right)_{\text{same branches}}(\mathbf{1})\ast t'=|\vee (t')^{1} +(t'')^1\vee t'_2 + t''\vee (t'_2)^1
\end{equation}
for each pair of leaves $(i, i+1)$ on the left or right branches. To collect all  different terms produced in this way is equivalent to sum over $(t')^1\in \left(Y^{n-1}\right)^1$, also over $t''\in Y^{a+1}$ and $(t'')^{1}\in \left(Y^{a+1}\right)^1$ for $0\le a\le n-2$, and finally over $t'_2\in Y^{b}$ and $(t'_2)^1\in \left(Y^{b}\right)^1$ for $0\le b\le n-2$. Then the sum of all different terms given by (\ref{eq:identifysepbr}) is\footnote{Assuming the convention that for a tree $t$ of order 0, $(t)^1=0$.}
\begin{align}\label{eq:identifysepbr2}
& \sum_{i=0}^{n-1}\sum_{t'\in Y^{n-1}}\left(_{i}\leftrightarrow_{i+1}\right)_{\text{same branches}}(\mathbf{1})\ast t'=\sum_{(t')^1\in \left(Y^{n-1}\right)^1} |\vee (t')^1\notag\\
& +\sum_{\substack{(t'')^{1}\in \left(Y^{a+1}\right)^1, t'_2\in Y^{b}\\a+b+1=n-1}}(t'')^1\vee t'_2 + \sum_{\substack{t^{''}\in Y^{a+1},(t'_2)^1\in \left(Y^{b}\right)^1\\a+b+1=n-1}}t''\vee (t'_2)^1\notag\\
&=\sum_{\substack{(t_1)^1\in \left(Y^p\right)^1, t_2\in Y^q\\p+q+1=n}}t_1^1 \vee t_2 +\sum_{\substack{t_1\in Y^p, t^1_2\in \left(Y^q\right)^1\\p+q+1=n}} t_1\vee t_2^1
\end{align}
where now $t_1$ and $t_2$ are the left and right branches of a tree $t=t_1\vee t_2$ of order $n$. Note that on the right the first sum starts at $p=1$ which is the lowest possibility for a $g=1$ graph. Then the highest value of $q$ is $n-2$. For the same reason the second sum on the right starts at $q=1$ which implies that $p\le n-2$. Translating (\ref{eq:identifysepbr2}) to the topological recursion we have 
\begin{align}
& \psi^\ast\left(\sum_{i=0}^{n-1}\sum_{t'\in Y^{n-1}}\left(_{i}\leftrightarrow_{i+1}\right)_{\text{same branches}}(\mathbf{1})\ast t'\right)=K_p(q,\bar{q})\times\notag\\
&\left(\sum_{L\cup M=\{p_1,\dots,p_{n-1}\}, |L|=1}W_2^0(q,L)W_{m+1}^1(\bar{q},M)\right.\notag\\
&\left.+\sum_{L\cup M=\{p_1,\dots,p_{n-1}\}, |L|>1}W_{l+1}^1(q,L)W_{m+1}^0(\bar{q},M)
+W_{l+1}^0(q,L)W_m^1(\bar{q},M)\right)\notag\\
&=K_p(q,\bar{q})\times\notag\\
&\left(\sum_{L\cup M=\{p_1,\dots,p_{n-1}\}}W_{l+1}^0(q,L)W_{m+1}^1(\bar{q},M)
+W_l^1(q,L)W_m^0(\bar{q},M)\right)
\end{align}
which is the second term of the topological recursion formula for $W_{n}^1(p,p_1,\dots,p_{n-1})$.
\end{proof}
Now we consider the first term in topological recursion which for $W_{n}^1(p,p_1,\dots,p_{n-1})$ is
\begin{equation}\label{eq:firstW^1}
K_p(q,\bar{q})W^0_{n+1}(q,\bar{q},p_1,\dots,p_{n-1}).
\end{equation} 
We start by a definition:
\begin{defn}
The ungrafting operation $\raisebox{1.7mm}{$\line(1,0){10}$}\hspace{-3mm}\vee$ is defined by removing from $t$ the tree $(\mathbf{1})$ that contains the root producing a forest with two trees $t_1$ and $t_2$. When $t$ represents an instance of a correlation function then the roots of $t_1$ and $t_2$ are labeled by $q$ and $\bar{q}$ and as before the tree $(\mathbf{1})$ represents $K_p(q,\bar{q})$.
\end{defn}
\begin{rem}
The operations $\vee$ and $\raisebox{1.7mm}{$\line(1,0){10}$}\hspace{-3mm}\vee$ are similar to the operations $B^+$ and $B^-$ of the Connes-Kreimer Hopf Algebra described, for instance, in \cite{MR1725011}.
\end{rem}
If we start with the planar binary tree $t\in Y^n$ with leaves labels $p_1,\dots, p_{n+1}$ and root label $p$ and identify two nearest neighbor leaves in opposite branches then we get a 1-loop graph $t^1\in \left(Y^n\right)^1$ with a relabeling $p_1,\dots,p_{n-1}$. Then, by applying $\raisebox{1.7mm}{$\line(1,0){10}$}\hspace{-3mm}\vee$, we get another tree $t'$ with two more edges with labels $q$ and $\bar{q}$ besides the leaves labelled by $p_1,\dots,p_{n-1}$. This tree is isomorphic as a graph to a planar binary tree in $Y^{n-1}$ that we denote also $t'$ by promoting the edge with the label $q$ to the root and the other edge to the rightmost leaf, see fig \ref{fig:ungraft} for an example with $W_3^1(p,p_1,p_2)$. In this way we get a representation of (\ref{eq:firstW^1}) by summing over all planar binary trees $t'\in Y^{n-1}$.
\begin{figure}
 \begin{tikzpicture}
     \draw[thick] (1,-1) -- (1,-0.5) -- (0.5,0) -- (1,-0.5) -- (1.5,0);
     \draw[thick] (0.2,1.15)--(0.7,1.65)--(0.5,1.45)--(0.3,1.65)--(0.5,1.45)--(0.2,1.15)--(-0.3,1.65)-- (0.5,0.85);  \draw[thick] (1.5,0.85)--(1.95,1.3)--(2.3,1.65);
     \draw (0.5,0.2) node{ $q$};\draw (1.5,0.2) node{ $\bar{q}$};
     \draw (0.5,0.6) node{ $q$};\draw (1.5,0.6) node{ $\bar{q}$};
     \draw (1,-1.2) node{ $p$};
      \draw (-0.35,1.85) node{ $p_1$};\draw (0.3,1.85) node{ $p_2$};
      \draw[thick] (2.3,1.7) arc (45:135:1.10cm);
      \draw (2.8,0.5) node{\textbf{{\Large $\longrightarrow$}}};
       \draw[thick] (4.5,-1) -- (4.5,-0.5) -- (4,0) -- (4.5,-0.5) -- (5,0);
       \draw (4.5,-1.2) node{ $p$};
       \draw (4,0.2) node{ $q$};\draw (5,0.2) node{ $\bar{q}$};
        \draw[thick] (6,0.85) -- (5.2,1.65) -- (6,0.85)--(6.8,1.65)--(6.5,1.35) -- (6.2,1.65)--(6.5,1.35)-- (6,0.85)--(6,0.35);
         \draw (6,0.15) node{ $q$};
         \draw (5.2,1.85) node{ $p_1$};\draw (6.2,1.85) node{ $p_2$};
         \draw (6.8,1.85) node{ $\bar{q}$};
      \end{tikzpicture} 
\caption{Ungrafting a 1-loop graph of $W_3^1(p,p_1,p_2)$. The resulting tree is $(\mathbf{21})$.}\label{fig:ungraft}
\end{figure}
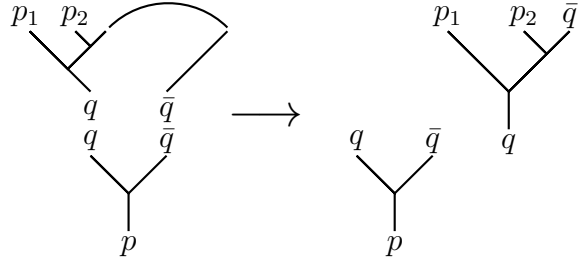
It is clear that what was left to be done in (\ref{eq:identifysepbr}) was the identification of two consecutive leaves in opposite branches. In principle there are several graphs $t^1\in(Y^n)^1$ of this type. The first one corresponds to the first term in (\ref{eq:shuffleident1a}) which gives a 1-loop graph with no leaves on the left branch. All other $t^1$ come from the second term of that formula and depend on the type of the left branch $t'_1$ of the decomposition of the tree $t'\in Y^{n-1}$, $t'=t'_1\vee t'_2$. If $t'_1=|$ then $(\mathbf{1})\ast |=(\mathbf{1})$ and the identification of leaves gives a $t^1$ with a single leaf on the left branch. The next case is $(\mathbf{1})\ast(\mathbf{1})=(\mathbf{12})+(\mathbf{21})$ and this gives a sum of two one-loop graphs with 2 leaves on the left branch. The procedure continues until $t'_2=|$ which then gives a 1-loop graph with no leaves on the right branch. Thus we have proved the following proposition
\begin{prop}\label{first-term-top-rec}
	The representation of the first term of the topological recursion formula for $W^1_{n}(p,p_1,\dots,p_{n-1})$ is given by the identification of leaves on opposite branches of the decomposition $t=t_1\vee t_2$ with $t\in Y^n, t_1\in Y^p, t_2\in Y^q,p+q+1=n$:
	\begin{equation}\label{eq:prop3}
	\psi\left(K_p(q,\bar{q})W^0_{n+1}(q,\bar{q},p_1,\dots,p_{n-1})\right)=\sum_{\substack{t_1\in Y^p, t_2\in Y^q\\p+q+1=n}}t_1 \bridge t_2 
	\end{equation}
	\end{prop}
The obvious notation $\phantom{.}\bridge\phantom{.}$ means that two consecutive leaves in opposite branches are identified.

Therefore we have exhausted all possibilities of obtaining 1-loop graphs from planar binary trees of order $n$ and the two previous propositions imply the following theorem:
\begin{thm}
The $n$ order solution $W^1_{n}(p,p_1,\dots,p_{n-1})$ of the topological recursion in genus $1$ is given by
$\mathbf{(1)}\ast\mathbf{(1)}\ast\dots\ast\mathbf{(1)}$, with $n$ factors, followed the identification of pairs of nearest neighbor leaves producing 1-loop graphs and finally by  summing over all permutations of leaves labels $p_1,p_2,\dots, p_{n-1}$:
\begin{align}
\psi\left(W^1_{n}(p,p_1,\dots,p_{n-1})\right)&=\notag\\
\sum_{\text{perm. of }\{p_1,\dots,p_{n-1}\}}\sum_{i=0}^{n-1} 
& _{i}\leftrightarrow_{i+1} \left(\mathbf{(1)}\ast\mathbf{(1)}\ast\dots\ast\mathbf{(1)}\right)
\end{align}
\end{thm}

Next we prove a simple lemma regarding symmetric graphs as in fig. \ref{fig:symgraph}. Note that the resulting ungrafted graphs have the left-right order of the right branch of the original graph exchanged:
\begin{lem}
If a graph that enters in the representation of the correlation function $W^{2g+1}_{k+1}(p,K)$ has nearest neighbor leaves identified in different branches and is symmetric with respect to the vertical axis that passes through the root then it has a weight factor of $2$, that is, it appears two times in the complete graph representation of $W^{2g+1}_{k+1}(p,K)$.
\begin{proof}
First note that such a graph has an even number of leaves, say $2a$, possibly 0. After being ungrafted the resulting graph represents the following term in the topological recursion formula:
\begin{align}
&K_p(q,\bar{q})W^{2g}_{2a+2}(q,\bar{q},K)\notag\\
&=K_p(q,\bar{q})K_q(q_1,\bar{q}_1)\left(W^{g}_{a+2}(q_1,\bar{q},L)W^{g}_{a+1}(\bar{q}_1,M)\right.\notag\\
&\left. +W^{g}_{a+1}(q_1,L)W^{g}_{a+2}(\bar{q}_1,\bar{q},M)+\dots\right)
\end{align}
for $|L|=|M|=a$ and $L\cup M =K$ where the dots represent other terms that are not symmetric and that are represented by other graphs. This shows that the same graph after being ungrafted originates two symmetrical graphs and as such it has a weight factor of 2 when counting the number of graphs in the complete representation.
\end{proof}
\end{lem}
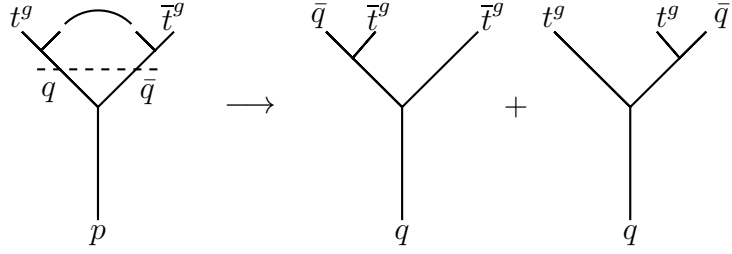
\begin{figure}
 \begin{tikzpicture}
 \draw[thick](1,-0.5) -- (1,1) --(0.25,1.75)--(0.5,2)--(0.25,1.75)-- (0,2) -- (1,1) -- (1.75,1.75)--(1.5,2)--(1.75,1.75)--(2,2) ; \draw (0,2.2) node{ $t^g$};
 \draw (2,2.2) node{ $\overline{t}^g$}; \draw[thick] (1.44,2.1) arc (45:135:0.62cm);
 \draw (3,1) node{ $\longrightarrow$};
 \draw[thick,dashed] (0.2,1.5)--(1.8,1.5);\draw (0.35,1.2) node{ $q$};
 \draw (1.65,1.2) node{ $\bar{q}$};
 \draw (1,-0.7) node{ $p$};
 \draw[thick] (5,-0.5) -- (5,1) -- (4,2) -- (4.35,1.65)--(4.65,2)--(4.35,1.65)-- (5,1) -- (6,2);\draw (6.5,1) node{ $+$}; \draw[thick] (8,-0.5) -- (8,1) -- (7,2) -- (8,1) -- (9,2)--(8.65,1.65)--(8.35,2)--(8.65,1.65);
  \draw (3.9,2.2) node{ $\bar{q}$};\draw (4.7,2.2) node{$\overline{t}^g$};\draw (6.2,2.2) node{ $\overline{t}^g$};
 \draw (7,2.2) node{ $t^g$};\draw (8.5,2.2) node{$t^g$};
  \draw (9.2,2.2) node{ $\bar{q}$};
  \draw (5,-0.7) node{ $q$}; \draw (8,-0.7) node{ $q$};
 \end{tikzpicture} 
\caption{A symmetric $2g+1$-loop graph with $t^g$ a $g-$loop graph. The graph $\overline{t}^g$ is the reflection of $t^g$ on the vertical axis that passes through the root. After the ungrafting operation the first graph is obtained by exchanging $q$ and $\bar{q}$ on the original graph.}\label{fig:symgraph}
\end{figure}
\begin{defn2}
See fig. \ref{fig:W13} for the graph representation of $W_3^1(p,p_1,p_2)$ with a graph of weight 2.
\end{defn2}
\begin{figure}[h]
	\begin{tikzpicture}
	\draw[thick] (1,-0.5) -- (1,1) -- (0,2) -- (0.2,1.8)--(0.4,2)--(0.2,1.8)-- (0.5,1.5)--(1,2)--(0.5,1.5)-- (1,1) -- (2,2);\draw (2.5,1) node{\textbf{{\Large $+$}}}; \draw[thick] (4,-0.5) -- (4,1) -- (3,2) -- (4,1) -- (5,2)--(4.8,1.8)--(4.6,2)--(4.8,1.8)--(4.5,1.5)--(4,2);\draw (5.5,1) node{\textbf{{\Large $+$}}}; \draw[thick] (7,-0.5) -- (7,1) -- (6,2) -- (7,1) -- (8,2)--(7.5,1.5)--(7.25,1.75)--(7.5,2)--(7.25,1.75)--(7,2);\draw[thick] (2,2.1) arc (45:135:0.67cm);\draw[thick] (3.95,2.1) arc (30:150:0.58cm);\draw[thick] (6.95,2.1) arc (30:150:0.58cm);
	\end{tikzpicture}
	\begin{tikzpicture}
	\draw (-0.5,1) node{\textbf{{\Large $+$}}};\draw[thick] (1,-0.5) -- (1,1) -- (0,2) -- (0.35,1.65)--(0.65,2)--(0.35,1.65)-- (1,1) --(1.65,1.65)--(1.35,2)--(1.65,1.65) -- (2,2);\draw (2.5,1);  \draw (3.5,1) node{\textbf{{\Large $+$}}}; \draw[thick] (7,-0.5) -- (7,1) -- (6,2) -- (6.5,1.5)--(6.75,1.75)--(6.5,2)--(6.75,1.75)--(7,2)-- (6.5,1.5) -- (7,1)-- (8,2);
	\draw[thick] (2,2.1) arc (30:150:0.35cm);\draw[thick] (8,2.1) arc (45:135:0.67cm);
	\end{tikzpicture} 
	\begin{tikzpicture}
	\draw[thick] (1,-0.5) -- (1,1) -- (0,2) -- (0.2,1.8)--(0.4,2)--(0.2,1.8)-- (0.5,1.5)--(1,2)--(0.5,1.5)-- (1,1) -- (2,2);\draw (2.5,1) node{\textbf{{\Large $+$}}}; \draw[thick] (4,-0.5) -- (4,1) -- (3,2) -- (4,1) -- (5,2)--(4.8,1.8)--(4.6,2)--(4.8,1.8)--(4.5,1.5)--(4,2);\draw (5.5,1) node{\textbf{{\Large $+$}}}; \draw[thick] (7,-0.5) -- (7,1) -- (6,2) -- (7,1) -- (8,2)--(7.5,1.5)--(7.25,1.75)--(7.5,2)--(7.25,1.75)--(7,2);\draw[thick] (0.4,2.1) arc (30:150:0.2cm);\draw[thick] (5,2.1) arc (45:135:0.3cm);\draw[thick] (8.05,2.1) arc (45:135:0.35cm);
	\end{tikzpicture}
	\begin{tikzpicture}
	\draw (-0.5,1) node{\textbf{{\Large $+$}}};\draw[thick] (1,-0.5) -- (1,1) -- (0,2) -- (0.35,1.65)--(0.65,2)--(0.35,1.65)-- (1,1) --(1.65,1.65)--(1.35,2)--(1.65,1.65) -- (2,2);\draw (2.5,1);  \draw (3.5,1) node{\textbf{{\Large $+$}}}; \draw[thick] (7,-0.5) -- (7,1) -- (6,2) -- (6.5,1.5)--(6.75,1.75)--(6.5,2)--(6.75,1.75)--(7,2)-- (6.5,1.5) -- (7,1)-- (8,2);
	\draw[thick] (0.6,2.1) arc (30:150:0.35cm);\draw[thick] (6.44,2.1) arc (30:150:0.29cm);
	\end{tikzpicture} 
	\begin{tikzpicture}
	\draw[thick] (1,-0.5) -- (1,1) -- (0,2) -- (0.2,1.8)--(0.4,2)--(0.2,1.8)-- (0.5,1.5)--(1,2)--(0.5,1.5)-- (1,1) -- (2,2);\draw (2.5,1) node{\textbf{{\Large $+$}}}; \draw[thick] (4,-0.5) -- (4,1) -- (3,2) -- (4,1) -- (5,2)--(4.8,1.8)--(4.6,2)--(4.8,1.8)--(4.5,1.5)--(4,2);\draw (5.5,1) node{\textbf{{\Large $+$}}}; \draw[thick] (7,-0.5) -- (7,1) -- (6,2) -- (7,1) -- (8,2)--(7.5,1.5)--(7.25,1.75)--(7.5,2)--(7.25,1.75)--(7,2);\draw[thick] (1,2.1) arc (45:135:0.4cm);\draw[thick] (4.6,2.1) arc (45:135:0.4cm);\draw[thick] (7.5,2.1) arc (45:135:0.35cm);
	\end{tikzpicture}
	\begin{tikzpicture}
	\draw (-0.5,1) node{\textbf{{\Large $+2$}}};\draw[thick] (1,-0.5) -- (1,1) -- (0,2) -- (0.35,1.65)--(0.65,2)--(0.35,1.65)-- (1,1) --(1.65,1.65)--(1.35,2)--(1.65,1.65) -- (2,2);\draw (2.5,1);  \draw (3.5,1) node{\textbf{{\Large $+$}}}; \draw[thick] (7,-0.5) -- (7,1) -- (6,2) -- (6.5,1.5)--(6.75,1.75)--(6.5,2)--(6.75,1.75)--(7,2)-- (6.5,1.5) -- (7,1)-- (8,2);
	\draw[thick] (1.27,2.1) arc (45:135:0.42cm);\draw[thick] (7,2.1) arc (30:150:0.29cm);\draw (3.5,-1) node{\textbf{$+$ perm. of $\{p_1,p_2\}$}};
	\end{tikzpicture} 
	\caption{$W_3^1(p,p_1,p_2)$. The root label $p$ and the leaf labels $\{p_1,p_2\}$ are omitted.}\label{fig:W13}
\end{figure}
A simple but important fact is that $_{i}\leftrightarrow_{i+1}$ acts as a derivation when applied independently to the left and right branches of a tree $t=t_1\vee t_2$. This is apparent in Proposition \ref{prop:secterm-W1}. However when summing over all graphs we must take care with overcounting. If we start with $t^1=(t_1)^1\vee t_2+t_1\vee (t_2)^1$ with $|t|=n, |t_1|=p,|t_2|=q, p+q+1=n$ and apply $_{i}\leftrightarrow_{i+1}$ to the two branches independently and sum over all different graphs we get
\begin{align}
&\sum_{i=0}^{n-1} \left(_{i}\leftrightarrow_{i+1}\right)_{\text{same branches}}\left(\sum_{\substack{(t_1)^1\in (Y^{p})^1,t_2\in Y^{q}\\p+q+1=n}}(t_1)^1\vee t_2\right.\notag\\
&\left.+\sum_{\substack{t_1\in Y^p, (t_2)^1\in (Y^{q})^1\\p+q+1=n}}t_1\vee (t_2)^1\right)=\notag\\
&\sum_{\substack{(t_1)^2\in (Y^{p})^2,t_2\in Y^{q}\\p+q+1=n}}2(t_1)^2\vee t_2
+\sum_{\substack{(t_1)^1\in (Y^{p})^1,(t_2)^1\in (Y^{q})^1\\p+q+1=n}}2(t_1)^1\vee (t_2)^1\notag\\
&+\sum_{\substack{t_1\in Y^p, t_2\in Y^{q},(t_2)^2\in (Y^{q})^2\\p+q+1=n}}2t_1\vee (t_2)^2
\end{align}
whenever the operation is well defined\footnote{The operation is not well defined if there is only one leaf available before and/or after a certain loop or if there are no more leaves to contract. In this case we set to 0 the result of acting with $_{i}\leftrightarrow_{i+1}$.}. The reasoning for the 2 factors is the double counting of identical graphs. For instance if $(t_1)^1$ has a loop starting at leaf $2$ and $(t_1)^2$ was obtained by producing a second loop starting at leaf 0, then this $(t_1)^2$ is identical to the 2-loop graph that was obtained by producing a second loop at leaf 2 in $(t_1)^1$ that had already a loop starting at leaf 0. 

Therefore the sum of different 2-loop graphs obtained from all planar binary trees $t=t_1\vee t_2$ is
\begin{align}
&\sum_{(t)^2\in (Y^n)^2} (t)^2=\sum_{\substack{(t_1)^2\in (Y^{p})^2,t_2\in Y^{q}\\p+q+1=n}}(t_1)^2\vee t_2\notag\\
&+
 \sum_{\substack{(t_1)^1\in (Y^{p})^1,(t_2)^1\in (Y^{q})^1\\p+q+1=n}}(t_1)^1\vee (t_2)^1\notag\\
 &+\sum_{\substack{t_1\in Y^p, (t_2)^2\in (Y^{q})^2\\p+q+1=n}}t_1\vee (t_2)^2
\end{align}

As for the identification of two consecutive leafs in separate branches, we start again from $t^1=(t_1)^1\vee t_2+t_1\vee (t_2)^1$ to get
\begin{equation}
t^2=(t_1)^1\bridge t_2+t_1\bridge (t_2)^1
\end{equation}
The fact that it may not be always possible to contract two leaves in opposite branches is important to count the dimensions of the vector spaces $k[(Y^n)^g]$ generated by graphs with $g$ loops but we will not consider this. 

It is clear that we can continue this procedure and generate graphs $t^g\in (Y^n)^g$ with an increasing number of $g$ loops and $k$ labels (including the root) up to the consistence of the relation $-n=2-2g-k$.
Hence we have the following proposition:
\begin{prop}\label{prop:graphloop}
The graphs with loops $t^g\in (Y^n)^g$ and $k$ labels, including the root, that are compatible with $-n=2-2g-k$ are obtained from planar binary trees $t=t_1\vee t_2$ by successive applications of $\left(_{i}\leftrightarrow_{i+1}\right)_{\text{same branches}}$ and $\left(_{i}\leftrightarrow_{i+1}\right)_{\text{opposite branches}}$ to all $t\in Y^n$:
\begin{align}\label{eq:gloopgraph}
\sum_{(t)^g\in (Y^n)^g} (t)^g=&\sum_{k=0}^g\sum_{\substack{(t_1)^k\in(Y^p)^k,(t_2)^{g-k}\in(Y^q)^{g-k}\\p+q+1=n}} \left( (t_1)^k\vee (t_2)^{g-k}\right)\notag\\
&+\sum_{k=0}^{g-1} \sum_{\substack{(t_1)^k\in(Y^p)^k,(t_2)^{g-k}\in(Y^q)^{g-k}\\p+q+1=n}}\left((t_1)^{g-1-k}\bridge (t_2)^{k}\right)
\end{align}
\end{prop}
\begin{proof}
It follows from the discussion above and a simple inductive argument. The first sum on the right is a consequence of the fact that $\left(_{i}\leftrightarrow_{i+1}\right)_{\text{s. br.}}$ acts as a derivation, after taken into account double counting and summing over all different graphs, and the second sum results from applying $\left(_{i}\leftrightarrow_{i+1}\right)_{\text{opp. br.}}$ to a $(g-1)-$loop graph $t^{g-1}$ without leafs from opposite branches identified. 
To see that any graph $t^g$ on the sum on the left can be obtained in this way just take the same graph but with some pair of leafs free, say $(i,i+1)$. This is a $t^{g-1}$ graph that admits a decomposition $(t_1)^k\vee (t_2)^{g-1-k}$ or $(t_1)^k\bridge (t_2)^{g-2-k}$. In the first case, if $(i,i+1)$ belong to the left or right branches apply $\left(_{i}\leftrightarrow_{i+1}\right)_{\text{s. br.}}$ to get an element in the first sum on the right of (\ref{eq:gloopgraph}). If $(i,i+1)$ belong to opposite branches apply $\left(_{i}\leftrightarrow_{i+1}\right)_{\text{opp. br.}}$ to get an element of the second sum. The second case works in a similar way except that $(i,i+1)$ must belong to the same branches.
\end{proof}

\begin{thm}
The $n$ order solution $W^g_{2-2g+n}$ of the topological recursion in genus $g>0$ and $k=2-2g+n>0$ variables is given by
$\mathbf{(1)}\ast\mathbf{(1)}\ast\dots\ast\mathbf{(1)}$, with $n$ factors, followed the identification of pairs of nearest neighbor leaves producing graphs with loops as in Proposition \ref{prop:graphloop} and finally by  summing over all permutations of $p_1,p_2,\dots, p_{1-2g+n}$.
\end{thm}
\begin{proof}
The proof now follows easily from Proposition \ref{prop:graphloop}. Remember that the sum of all planar binary trees of order $n$ is obtained by $n$ factors of $(\mathbf{1})$ with the $\ast$ product. By ungrafting the graphs given by (\ref{eq:gloopgraph}) and using the representation map $\psi$ we get
\begin{align}\label{eq:prooftoprec}
&W^g(p,p_1,\dots,p_{k-1})=\psi^\ast\left(\sum_{\substack{(t)^g\in (Y^n)^g\\\text{perm. of leaf labels }K=\{p_1,\dots,p_{k-1}\}}}(t)^g\right)=\notag\\
&\sum_{h=0}^g\sum_{\substack{L\cup M=K\\\text{perm. of }K}}
 K_p(q,\bar{q})\left(W^{g-1}_{k+1}(q,\bar{q},K)+ W_{l+1}^h(q,L)W_{m+1}^{g-h}(\bar{q},M)\right)
\end{align}
which is the topological recursion formula for arbitrary genus.
\end{proof}
\begin{rem}
After being ungrafted the second term in (\ref{eq:gloopgraph}) can still have leaves in opposite branches identified. In general this happens if $(t)^g=(t_1)^{g_1}\bridge (t_2)^{g_2}$ with $g_1+g_2=g-1$ and say $(t_1)^{g_1}$ has a decomposition $(t_1)^{g'_1}\bridge (t_2)^{g'_2}$ with ${g'}_1+{g'}_2=g_1-1$.
 \begin{defn2}
 See fig. \ref{fig:W21} for the graph representation of $W_1^2(p)$.
 \end{defn2}
  \begin{figure}[h]
 	\begin{tikzpicture}
 	\draw[thick] (1,-0.5) -- (1,1) -- (0,2) -- (0.2,1.8)--(0.4,2)--(0.2,1.8)-- (0.5,1.5)--(1,2)--(0.5,1.5)-- (1,1) -- (2,2);\draw (2.5,1) node{\textbf{{\Large $+$}}}; \draw[thick] (4,-0.5) -- (4,1) -- (3,2) -- (4,1) -- (5,2)--(4.8,1.8)--(4.6,2)--(4.8,1.8)--(4.5,1.5)--(4,2);\draw (5.5,1) node{\textbf{{\Large $+$}}}; \draw[thick] (7,-0.5) -- (7,1) -- (6,2) -- (7,1) -- (8,2)--(7.5,1.5)--(7.25,1.75)--(7.5,2)--(7.25,1.75)--(7,2);\draw[thick] (0.4,2.1) arc (30:150:0.2cm);\draw[thick] (2,2.1) arc (45:135:0.67cm);\draw[thick] (3.95,2.1) arc (30:150:0.58cm);\draw[thick] (5,2.1) arc (45:135:0.3cm);\draw[thick] (6.95,2.1) arc (30:150:0.58cm);\draw[thick] (8.05,2.1) arc (45:135:0.35cm);
 	\end{tikzpicture}
 	\begin{tikzpicture}
 	\draw (-0.5,1) node{\textbf{{\Large $+$}}};\draw[thick] (1,-0.5) -- (1,1) -- (0,2) -- (0.35,1.65)--(0.65,2)--(0.35,1.65)-- (1,1) --(1.65,1.65)--(1.35,2)--(1.65,1.65) -- (2,2);\draw (2.5,1);  \draw (3.5,1) node{\textbf{{\Large $+$}}}; \draw[thick] (7,-0.5) -- (7,1) -- (6,2) -- (6.5,1.5)--(6.75,1.75)--(6.5,2)--(6.75,1.75)--(7,2)-- (6.5,1.5) -- (7,1)-- (8,2);
 	\draw[thick] (0.6,2.1) arc (30:150:0.35cm);\draw[thick] (2,2.1) arc (30:150:0.35cm);\draw[thick] (6.44,2.1) arc (30:150:0.29cm);\draw[thick] (8,2.1) arc (45:135:0.67cm);
 	\end{tikzpicture} 
 	\caption{$W_1^2(p)$ The root label $p$ is omitted.}\label{fig:W21}
 \end{figure}
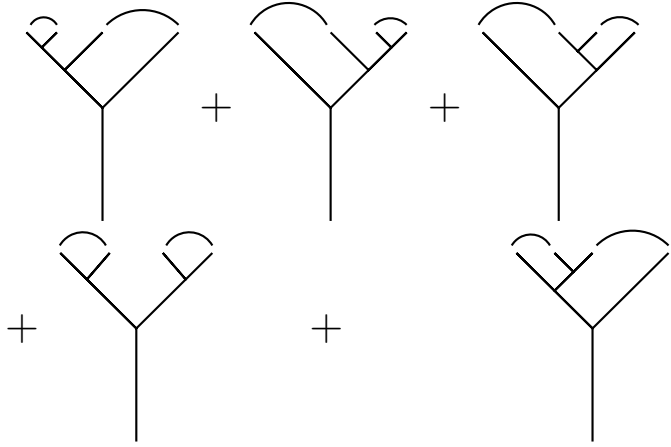
\end{rem}
\section{The antipode}
In a graded connected Hopf Algebra there is a canonical antipode $S$ whose expression is given by the convolution inverse of the identity:
\begin{equation}
m\left(S\otimes \text{Id}\right)\Delta = m\left(I\otimes \text{S}\right)\Delta=\eta\cdot\epsilon
\end{equation}
with $m$ the product, $\Delta$ the co-product, $\eta$ the unit and $\epsilon$ the co-unit. Explicitly, in the Loday-Ronco Hopf Algebra, we have
\begin{equation}
S(t)=-t-S(t_1)\ast t_2
\end{equation}
where in Sweedler notation
\begin{equation}
\Delta t = \sum t_1\otimes t_2
\end{equation}
is the co-product in $k[Y^\infty]$ induced by (\ref{eq:coproductperm}).
For instance, $S(\mathbf{1})=-\mathbf{1}$ because $\mathbf{1}$ is primitive and $S(\mathbf{12})=(\mathbf{21})$ and also $S(\mathbf{21})=(\mathbf{12})$. This suggests that a map $\psi^\ast S$ induced by the antipode on the vector space of correlations functions should give 
$$(\psi^\ast S)(W^0_4)=W^0_4.$$
More generally, from $S((\mathbf{1}))=-(\mathbf{1})$ we see that
$$S((\mathbf{1})\ast(\mathbf{1})\ast\dots\ast(\mathbf{1}))=(-1)^n(\mathbf{1})\ast(\mathbf{1})\ast\dots\ast(\mathbf{1})$$
with $n$ factors in the $\ast$ product and then $$(\psi^\ast S)(W_{n+2}^0)=(-1)^nW_{n+2}^0.$$
Since $n$ is identified with the Euler characteristic we see that the induced map respects the grading of $k[Y^\infty]$ for $g=0$.
For the moment it is not clear how to extend this simple computation to the case of graphs with loops. The tree |  being the identity in $k[Y^\infty]$ and representing $W_2^0$ has trivially $S(|)=|$, but it is not clear if the change of topology from a tree to a graph with loops shouldn't change dramatically the Hopf Algebra structure or if even the full algebra of these class of graphs with loops of arbitrary order is yet an Hopf Algebra. There are examples of Hopf Algebra of general graphs that are well documented in the literature (see for instance \cite{MR2967484,MR1303288}) but we do not know at the moment if they can be adapted to this framework. In particular, the Euler characteristic of a graph has a different meaning of the one used here.

\section{Discussion}
We have seen that an extension of the Hopf algebra planar binary trees of Loday and Ronco provides a representation of a vector space, whose nature is still to be clarified, generated by the set of correlations functions typical from Matrix Models and that satisfy the recursion formula of Eynard and Orantin. This extension, obtained by identifying nearest neighbor leaves through a single edge, is only necessary for $g>0$. This procedure moves from planar binary trees to planar graphs with loops, binary in the internal vertices in the sense that each internal vertex has two children. In the process of showing that this class of graphs satisfies the full recursion formula of Eynard-Orantin we have provided an explicit formula for the solutions, that are obtained first by getting the sum of all planar binary trees of order $n$ through the $\ast$ product of $n$ factors of $\mathbf{(1)}$ computed in $k[Y^\infty]$ and then by connecting two consecutive leaves in all possible ways up to the genus $g$.

\bibliographystyle{amsplain}
\bibliography{biblioHopf}

\providecommand{\bysame}{\leavevmode\hbox to3em{\hrulefill}\thinspace}
\providecommand{\MR}{\relax\ifhmode\unskip\space\fi MR }
\providecommand{\MRhref}[2]{%
  \href{http://www.ams.org/mathscinet-getitem?mr=#1}{#2}
}
\providecommand{\href}[2]{#2}
\begin{thebibliography}{10}

\bibitem{MR2194965}
Marcelo Aguiar and Frank Sottile, \emph{Structure of the {L}oday-{R}onco {H}opf
  algebra of trees}, J. Algebra \textbf{295} (2006), no.~2, 473--511.

\bibitem{atiyah1988topological}
Michael~F Atiyah, \emph{Topological quantum field theory}, Publications
  Math{\'e}matiques de l'IH{\'E}S \textbf{68} (1988), 175--186.

\bibitem{Brouder:1999gk}
C.~Brouder, \emph{{Runge-Kutta methods and renormalization}}, Eur.Phys.J.
  \textbf{C12} (2000), 521--534.

\bibitem{Brouder:1999za}
Christian Brouder, \emph{{On the trees of quantum fields}}, Eur.Phys.J.
  \textbf{C12} (2000), 535--549.

\bibitem{collins1984renormalization}
J.C. Collins, \emph{Renormalization: An introduction to renormalization, the
  renormalization group and the operator-product expansion}, Cambridge
  University Press, 1984.

\bibitem{MR1725011}
A.~Connes and D.~Kreimer, \emph{Hopf algebras, renormalization and
  noncommutative geometry}, Quantum field theory: perspective and prospective
  ({L}es {H}ouches, 1998), NATO Sci. Ser. C Math. Phys. Sci., vol. 530, Kluwer
  Acad. Publ., Dordrecht, pp.~59--108.

\bibitem{MR1748177}
Alain Connes and Dirk Kreimer, \emph{Renormalization in quantum field theory
  and the {R}iemann-{H}ilbert problem. {I}. {T}he {H}opf algebra structure of
  graphs and the main theorem}, Comm. Math. Phys. \textbf{210} (2000), no.~1,
  249--273.

\bibitem{MR1810779}
\bysame, \emph{Renormalization in quantum field theory and the
  {R}iemann-{H}ilbert problem. {II}. {T}he {$\beta$}-function, diffeomorphisms
  and the renormalization group}, Comm. Math. Phys. \textbf{216} (2001), no.~1,
  215--241.

\bibitem{DiFrancesco:1993nw}
P.~Di~Francesco, Paul~H. Ginsparg, and Jean Zinn-Justin, \emph{{2-D Gravity and
  random matrices}}, Phys.Rept. \textbf{254} (1995), 1--133.

\bibitem{MR3087960}
Olivia Dumitrescu, Motohico Mulase, Brad Safnuk, and Adam Sorkin, \emph{The
  spectral curve of the {E}ynard-{O}rantin recursion via the {L}aplace
  transform}, Algebraic and geometric aspects of integrable systems and random
  matrices, Contemp. Math., vol. 593, Amer. Math. Soc., Providence, RI, 2013,
  pp.~263--315.

\bibitem{MR2346575}
B.~Eynard and N.~Orantin, \emph{Invariants of algebraic curves and topological
  expansion}, Commun. Number Theory Phys. \textbf{1} (2007), no.~2, 347--452.

\bibitem{MR1905177}
L.~Foissy, \emph{Les alg\`ebres de {H}opf des arbres enracin\'es d\'ecor\'es.
  {I}}, Bull. Sci. Math. \textbf{126} (2002), no.~3, 193--239.

\bibitem{MR1909461}
\bysame, \emph{Les alg\`ebres de {H}opf des arbres enracin\'es d\'ecor\'es.
  {II}}, Bull. Sci. Math. \textbf{126} (2002), no.~4, 249--288.

\bibitem{MR1817703}
Alessandra Frabetti, \emph{Simplicial properties of the set of planar binary
  trees}, J. Algebraic Combin. \textbf{13} (2001), no.~1, 41--65.

\bibitem{MR1327096}
Israel~M. Gelfand, Daniel Krob, Alain Lascoux, Bernard Leclerc, Vladimir~S.
  Retakh, and Jean-Yves Thibon, \emph{Noncommutative symmetric functions}, Adv.
  Math. \textbf{112} (1995), no.~2, 218--348.

\bibitem{MR2967484}
Brandon Humpert and Jeremy~L. Martin, \emph{The incidence {H}opf algebra of
  graphs}, SIAM J. Discrete Math. \textbf{26} (2012), no.~2, 555--570.

\bibitem{MR1654173}
Jean-Louis Loday and Mar{\'{\i}}a~O. Ronco, \emph{Hopf algebra of the planar
  binary trees}, Adv. Math. \textbf{139} (1998), no.~2, 293--309.

\bibitem{MR1303288}
William~R. Schmitt, \emph{Incidence {H}opf algebras}, J. Pure Appl. Algebra
  \textbf{96} (1994), no.~3, 299--330.

\end{thebibliography}
\end{document}